%% file: bp_voting.arxiv2.tex
\documentclass[letter]{article}
\usepackage[utf8]{inputenc}

\usepackage{aaai20}
\usepackage{times}
\usepackage{xcolor}
\usepackage{soul}
\usepackage[utf8]{inputenc}
\usepackage[small]{caption}

\usepackage{mathtools}
\usepackage{amssymb}
\usepackage{amsmath}
\usepackage{tikz}
\usetikzlibrary{calc}
\usepackage{algorithm}
\usepackage[noend]{algpseudocode}
\usepackage{todonotes}
\usepackage{dsfont}
\usepackage{enumitem}
\usepackage{tikz}
\usepackage{amsthm}
\usepackage{natbib}
\usepackage{thm-restate}
\usepackage{booktabs}
\usepackage{multirow}
\usepackage{newfloat}
\usepackage{wrapfig}
\usepackage{dsfont}
\usepackage{comment}
\usepackage{multirow}

\usepackage{mathtools}

\newtheorem{theorem}{Theorem}

\newtheorem{lemma}[theorem]{Lemma}
\newtheorem{definition}{Definition}
\newtheorem{example}{Example}

\DeclareMathOperator*{\argmax}{arg\,max}

\DeclareFloatingEnvironment{polytope}

\newcommand{\s}{\mathbf{s}}
\newcommand{\ct}{\mathbf{c}}
\newcommand{\p}{\mathbf{p}}

\title{Persuading Voters: It's Easy to Whisper, It's Hard to Speak Loud}
\author{
	Matteo Castiglioni, Andrea Celli, Nicola Gatti\\
	Politecnico di Milano, Piazza Leonardo da Vinci 32, I-20133, Milan, Italy\\
	\texttt{\{name.surname\}@polimi.it}\\
}

\date{\today}

\begin{document}

\maketitle

\begin{abstract}
	We focus on the following natural question: \emph{is it possible to influence the outcome of a voting process through the strategic provision of information to voters who update their beliefs rationally?}
	%
	%
	We investigate whether it is computationally tractable to design a signaling scheme maximizing the probability with which the sender's preferred candidate is elected.
	We resort to the model recently introduced by~\cite{arieli2019private} (\emph{i.e.}, without inter-agent externalities), and focus on, as illustrative examples, $k$-voting rules and plurality voting.
	There is a sharp contrast between the case in which \emph{private} signals are allowed and the more restrictive setting in which only \emph{public} signals are allowed.
	In the former, we show that an optimal signaling scheme can be computed efficiently both under a $k$-voting rule and plurality voting. 
	%
	%
	%
	%
	In establishing these results, we provide two contributions applicable to general settings beyond voting.
	Specifically, we extend a well-known result by~\cite{dughmi2017algorithmic} to more general settings and prove that, when the sender's utility function is anonymous, computing an optimal signaling scheme is fixed-parameter tractable in the number of receivers' actions.
	In the public signaling case, we show that the sender's optimal expected return cannot be approximated to within \emph{any} factor under a $k$-voting rule.
	This negative result easily extends to plurality voting and problems where utility functions are anonymous.
\end{abstract}



\input{intro}

\input{related_works}

\input{model}

\input{private_k}

\input{private_general}

\input{private_plurality}

\input{public}

\input{conclusion}
\clearpage

\begin{small}
\bibliographystyle{aaai}
\bibliography{biblio,dairefs}
\end{small}

\clearpage
\input{appendix}
\end{document}

%% file: intro.tex
\section{Introduction}

Information is the foundation of any democratic election, as it allows voters for better choices.
In many settings, uninformed voters have to rely on inquiries of third party entities to make their decision.
For example, in most trials, jurors are not given the possibility of choosing which tests to perform during the investigation or which questions are asked to witnesses. 
They have to rely on the prosecutor's investigation and questions.
The same happens in elections, in which voters gather information from third-party sources.
With the advent of modern media environments, malicious actors have unprecedented opportunities to garble this information and influence the outcome of the election via misinformation and fake news~\citep{allcott2017social}.
Reaching voters with targeted messages has never been easier.
Hence, we pose the question: \emph{can a malicious actor influence the outcome of a voting process only by the provision of information to voters who update their beliefs rationally?}

We describe the problem through the \emph{Bayesian persuasion} framework by~\cite{kamenica2011bayesian}.
At its core, the model involves an informed \emph{sender} trying to influence the behavior of self-interested \emph{receivers} through the provision of payoff-relevant information.
\cite{kamenica2011bayesian} study how a single sender, having access to some private information, can design a signaling scheme to persuade a single receiver to select a favorable action.
The model assumes the sender's commitment, which is realistic in many settings~\citep{kamenica2011bayesian,dughmi2017survey}.
One argument to that effect is that reputation and credibility may be a key factor for the long-term utility of the sender~\citep{rayo2010optimal}.

A number of recent works study \emph{social influence} as a means of election control~\citep{sina2015adapting,faliszewski2018opinion,wilder2018controlling,wilder2019defending}.
The crucial difference in our model is that voters are strategic players, who update their beliefs rationally.
This property forces the sender to carefully craft the signaling scheme to preserve persuasiveness (\emph{i.e.}, incentive compatibility).
Other mechanisms for election interference that have been studied are \emph{bribery}~\citep{faliszewski2009llull,erdelyi2017complexity}, and \emph{adding}/\emph{deleting voters}/\emph{candidates}~\citep{loreggia2015controlling,faliszewski2011multimode,liu2009parameterized,chen2017elections}.
Both these mechanisms differ from ours in the provision of tangible incentives (the former) or the modification of the election setting (the latter).
In our model, the malicious actor (\emph{i.e.}, the sender), can influence the outcome of the election only by deciding \emph{who gets to know what}.

We extend the fundamental model of~\cite{arieli2019private} to describe general voting problems.
Specifically, our model comprises of multiple receivers, an arbitrary number of actions and states of nature, and no inter-agent externalities.~\footnote{
%
Without inter-agent externalities, the utility of a voter depends only on her appreciation of the candidate she selected, and on the state of nature. 
It does not depend on actions taken by other voters.
}
We adopt the sender's perspective, and study how to compute optimal \emph{private} (\emph{i.e.}, different voters may receive different information) and \emph{public} (\emph{i.e.}, all voters observe the same signal) communication schemes.
We observe a sharp contrast between the two settings: all results for the private setting are positive (\emph{i.e.}, polynomial time tractability),  while coordinating voters via public signaling is largely intractable.

\textbf{Original contributions}.
We focus on two commonly adopted voting rules: \emph{$k$-voting rules} and \emph{plurality voting}.
First, we provide an efficient implementation of the optimal private signaling scheme under a $k$-voting rule.
Then, by generalizing a result by~\cite{dughmi2017algorithmic}---later revisited by~\cite{xu2019tractability}---, we describe a necessary and sufficient condition for the efficient computation of private signaling schemes for a general class of sender's objective functions.
This condition is employed to show that private Bayesian persuasion is fixed-parameter tractable in the number of receivers' actions when sender's utility is anonymous, and to show that an optimal private signaling scheme under plurality voting can be found in polynomial time.
As for public signaling, we provide a new inapproximability result on the problem of computing an optimal public signaling scheme under a $k$-voting rule, showing that it cannot be approximated to within \emph{any} factor of the input size.
This result significantly improves previous hardness results for this setting due to~\cite{dughmi2017algorithmic} and easily extends to anonymous sender's utility functions and plurality voting.

%% file: related_works.tex
\vspace{-.cm}
\section{Related Works}
\vspace{-.cm}
The classic model of Bayesian persuasion is due to~\cite{kamenica2011bayesian}. 
Later, \cite{bergemann2016bayes,bergemann2016information,bergemann2019information} highlighted the connection between optimal information disclosure and the \emph{best} Bayes correlated equilibrium from the sender's perspective.
A number of works deal with the multiple receivers generalization of the model, \emph{e.g.},~\cite{schnakenberg2015expert,taneva2015information,wang2013bayesian}.
Among these works, those by~\cite{bardhi2018modes},~\cite{alonso2016persuading} and~\cite{chan2019pivotal} are closely related to ours, representing the first attempts of applying the Bayesian persuasion framework to voting problems.
In particular,~\cite{bardhi2018modes} and~\cite{chan2019pivotal} focus on problems with binary actions and state spaces.
The former studies private Bayesian persuasion for unanimity voting, while the latter analyzes private and public persuasion with $k$-majority voting rules.
\cite{alonso2016persuading} employ a novel geometric tool to characterize an optimal public signaling scheme in a voting framework and characterize voters' preferences over electoral rules.
%
%
However,  the aforementioned works only provide the economic groundings of Bayesian persuasion in (simple) voting settings. Their characterization does not include any computational result.

\cite{dughmi2016algorithmic} are the first to analyze Bayesian persuasion from a computational perspective, in the single-receiver case.
\cite{arieli2019private} introduce the fundamental model of persuasion with no inter-agent externalities.~\footnote{
Problems with externalities are largely intractable even in very simple settings; see, e.g.,~\citep{bhaskar2016hardness,rubinstein2015honest,dughmi2018hardness}.
The no-externality assumption removes the equilibrium selection and (partially) the computation concerns arising in more general setting, and allows one to focus on the problem of coordinating receivers' actions.
} 
In the case of binary actions and state spaces, \cite{arieli2019private} provide an explicit characterization of the optimal private signaling scheme when the sender's utility function is either supermodular, anonymous submodular, or supermajority.
Moreover, \cite{arieli2019private} also provide necessary and sufficient conditions for the existence of a public signaling scheme matching the performance of the best private signaling scheme.
Other related works focusing on the no externality setting with receivers' binary action spaces are the following:
\cite{babichenko2016computational} describe a tight $(1-1/e)$-approximate private signaling scheme for monotone submodular sender's utility functions and show that an optimal private scheme for anonymous utility functions can be found in polynomial time.
%
%
\cite{dughmi2017algorithmic} generalize the model to the case of many states of nature, while assuming sender's utility to be a monotone set function.
In this setting, the supermodular and anonymous cases can be efficiently solved.
When the sender's utility is submodular, ~\cite{dughmi2017algorithmic} show that an $(1-1/e)$-approximation to the optimal revenue can be obtained by sending conditionally independent private signals to receivers.
Moreover, the authors show that it is \textsf{NP}-hard to approximate the sender's value provided by the optimal public scheme, within any constant factor.
Finally,~\cite{xu2019tractability} focuses on the complexity of public signaling when there are no inter-agent externalities and the action spaces are binary. 
Finding an optimal public signal is shown to be \emph{fixed-parameter tractable} under some non-degeneracy assumptions.
The author describes a PTAS with a bi-criteria guarantee for (monotone) submodular sender's objectives.


%% file: model.tex
\section{Model}

Our model is a generalization of the fundamental special case introduced by~\cite{arieli2019private}.
It comprises a \emph{sender} and a finite set $R$ of \emph{receivers} (voters) that must choose one alternative from a set $C=\{c_0,\ldots,c_\ell\}$ of candidates (\emph{i.e.}, $C$ is the set of voters's available actions).
Each voter must choose a candidate from $C$. 
Each voter's utility depends only on her own action and the (random) state of nature, but not on the actions of other voters.
In particular, we write $u_r: C\times \Theta\to\mathbb{R}$, where $\Theta=\{\theta_i\}_{i=1}^{n}$ is the finite space of states of nature.
The value of $u_r(c,\theta)$ is a measure of voter $r$'s appreciation of candidate $c$ when the state of nature is $\theta$.
A profile of votes (\emph{i.e.}, one candidate for each voter) is denoted by $\mathbf{c}\in\times_{r\in R} C$.
In general settings, beyond voting, we denote the sender's utility when the state of nature is $\theta$ with $f_\theta:\times_{r\in R}C\to \mathbb{R}$ (here $C$ may be an arbitrary space of actions).
Furthermore, we say that $f$ is \emph{anonymous} if its value depends only on $\theta$ and on the number of players selecting each action.
In the specific context of voting, the sender's objective is maximizing the winning probability of $c_0$ (according to some voting rules).
In this setting, instead of using $f$, we denote the sender's utility function by $W: \times_{r\in R} C\to\{0,1\}$, where $W(\cdot)=1$ if $c_0$ \emph{wins}, and $W(\cdot)=0$ otherwise.
The state of nature influences the receivers' preferences but it does not affect the sender's payoff, which only depends on the final votes.
~\footnote{The sender's utility function is state-independent in many settings, \emph{e.g.}, voting~\citep{alonso2016persuading}, and marketing~\citep{candogan2019persuasion,babichenko2016computational}.}

As it is customary in Bayesian persuasion (see, \emph{e.g.}, \cite{kamenica2018bayesian}), we assume $\theta$ is drawn from a common prior $\mu\in\textnormal{int}(\Delta(\Theta))$, which is explicitly known to both sender and receivers.~%
\footnote{$\textnormal{int}(X)$ is the \emph{interior} of set $X$, and $\Delta(X)$ is the set of all probability distributions on $X$.}
Their interaction goes as follows (see Figure~\ref{fig:time_line}): the sender commits to a publicly known \emph{signaling scheme} $\phi$ that maps states of nature to \emph{signals} for the voters.
\begin{figure}
	\hspace{-.5cm}
	\includegraphics[scale=1.1]{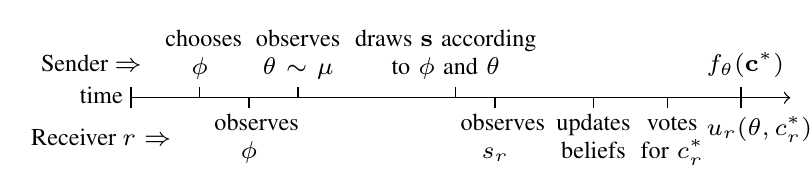}
	\caption{Interaction between the sender and a receiver.}
	\label{fig:time_line}
	\vspace{-.5cm}
\end{figure}
The \emph{signal set} of receiver $r$ is denoted by $S_r$, while $s_r \in S_r$ is a single signal for $r$. 
A profile of signals is denoted by $\mathbf{s}=(s_1,\ldots,s_{|R|})\in S$, where $S=\times_{r\in R} S_r$.
The sender observes the realized state of nature $\theta\sim\mu$, and exploits the knowledge of $\theta$ to compute 
 $\s\in S$, drawn according to $\phi$ under $\theta$.
Each voter $r$ observes $s_r$ and rationally updates her posterior over $\Theta$.
Then, each voter $r$ selects a candidate $c_r^\ast$ maximizing $r$'s expected reward.

A general signaling scheme $\phi:\Theta\to \Delta(S)$ has a private communication channel for each receiver. 
We refer to these as \emph{private signaling schemes}.
The probability with which the sender recommends $\mathbf{s}\in S$ after observing $\theta$ is $\phi(\theta,\s)$.
Therefore, it holds $\sum_{\s\in S}\phi(\theta,\s)=1$ for each $\theta$.
Let $S_{-r}= \times_{j\neq r} S_j$.
Then, $\phi_r(\theta,s_r)=\sum_{\s_{-r}\in S_{-r}} \phi(\theta,(s_r,\s_{-r}))$ is the \emph{marginal probability} with which $s_r\in S_r$ is recommended to $r$ when $\theta$ is observed.
Furthermore, we call $\phi_r$ a \emph{marginal signaling scheme}.
It is immediate to observe that the same set $\{\phi_r\}_{r\in R}$ may be obtained from more than one distribution over $S$.
A special setting is the one of \emph{public signaling schemes}, where the same signal is received by all receivers through a public communication channel (equivalently, all voters have to receive the same private message).
With an overload of notation, we denote a public signaling scheme by $\phi:\Theta\to S$, where $S$, denotes the set of public signals (the meaning of $\phi$ will be clear from the context).

A revelation-principle style argument shows that, in the case of no inter-agent externalities, there exists an optimal (private or public) signaling scheme which is both \emph{direct} and \emph{persuasive}~\citep{kamenica2011bayesian,dughmi2017algorithmic}.~\footnote{
An optimal signaling scheme always exists when the sender's expected utility is upper semicontinuous, which is always the case when receivers break ties in favor of the sender~\citep{kamenica2011bayesian}.
}
A signaling scheme is direct when signals can be mapped to actions of the receivers, and interpreted as action recommendations, \emph{i.e.}, in the voting setting, $S_r=C$ for each $r\in R$.
We say that a signaling scheme is \emph{persuasive} if following recommendations is an equilibrium of the underlying \emph{Bayesian game}~\citep{bergemann2016bayes,bergemann2016information}.
Analogously, a \emph{marginal signaling scheme} $\phi_r$ is persuasive if $r$ has no incentives in deviating from its recommendations.
When not specified, we assume to work with direct signaling schemes.
Moreover, when the sender's signaling scheme $\phi$ is direct and persuasive we write $W(\phi)$ to denote the sender's expected utility.
Finally, function $\delta: S\times C\to \mathbb{N}$ is s.t. $\delta(\s,c)$ is the number of voters that are recommended $c$ by $\s$.


To further clarify the notion of Bayesian persuasion in voting problems, we provide the following simple example.
\begin{example}
	There are three voters $R=\{1,2,3\}$ that have to choose one among two candidates $C=\{c_0,c_1\}$.
	The sender (\emph{e.g.}, a politician or a lobbyist) is interested in having $c_0$ elected, and can observe the realized state of nature, drawn uniformly from $\Theta=\{A,B,C\}$.
	The state of nature describes the position of $c_0$ on a matter of particular interest to the voters.
	Moreover, voters perceive $c_1$ in a slightly negative way, independently of the state of nature.
	Table~\ref{tab:example} describes the utility of the three voters.
	%

	\renewcommand{\arraystretch}{0.5}\setlength{\tabcolsep}{1pt}
	\begin{table}[h]
	\begin{minipage}{.49\linewidth}
			\scriptsize
			\centering
			\begin{tabular}{rl||ccr|lccr|lccr}
				\toprule
				&&& \multicolumn{2}{c}{State $A$} &&& \multicolumn{2}{c}{State $B$} &&& \multicolumn{2}{c}{State $C$} \\
				&&& $c_0$ & $c_1$ &&& $c_0$ & $c_1$ &&& $c_0$ & $c_1$ \\
				\midrule
				\multirow{3}{*}{\rotatebox[origin=c]{90}{Voters}} &1 && $+1$ & $-1/4$ &&& $-1$ & $-1/4$ &&& $-1$ & $-1/4$\\
				&2 && $-1$ & $-1/4$ &&& $+1$ & $-1/4$ &&& $-1$ & $-1/4$\\
				&3 && $-1$ & $-1/4$ &&& $-1$ & $-1/4$ &&& $+1$ & $-1/4$\\
				\bottomrule
			\end{tabular}
		\caption{Payoffs from voting different candidates.}
		\label{tab:example}
	\end{minipage}
	\hspace{1cm}
	\begin{minipage}{.3\linewidth}
			\centering
			\scriptsize
		\begin{tabular}{lr||lcr lcr lcr}
				\toprule
				&& \multicolumn{8}{c}{Signals}&\\
				&&& \textsf{not A} &&& \textsf{not B} &&& \textsf{not C}&\\
				\midrule
				\multirow{3}{*}{\rotatebox[origin=c]{90}{States}}&$A$ && 0 &&& $1/2$ &&& $1/2$ &\\
				&$B$ && $1/2$ &&& 0 &&& $1/2$ &\\
				&$C$ && $1/2$ &&& $1/2$ &&& 0 &\\
				\bottomrule
		\end{tabular}
			\caption{Optimal signaling scheme.}
			\label{tab:example_scheme}
	\end{minipage}
	\end{table}
	We consider a simple-majority voting rule. 
	Without any form of signaling (or with a completely uninformative signal), all voters would choose $c_1$ because it provides an expected utility of $-1/4$, against $-1/3$.
	This would clearly be the worst possible outcome for the sender, \emph{i.e.}, $W((c_1,c_1,c_1))=0$.
	The sender would still get 0 utility with a fully informative (private) signal, since two out of three receivers would pick $c_1$. 
	However, the sender can design a public signaling scheme granting herself utility 1 for each state of nature.~\footnote{Since a public signaling scheme is a special case of private signaling, the same result could be achieved in the latter scenario.}
	An optimal signaling scheme should convince two out of the three receivers to vote for $c_0$. 
	Table~\ref{tab:example_scheme} describes one such scheme with generic signals.
	%
	%
	Suppose the observed state is $A$, and that the signal is \textsf{not B} (sampled uniformly from $\{\textsf{not B},\textsf{not C}\}$).
	Then, the posterior distribution over the states of nature is $(1/2, 0, 1/2)$. 
	Therefore, both player 1 and 3 would vote for $c_0$, since their expected utility would be 0 against $-1/4$.
	The sender's payoff would be $W((c_0,c_1,c_0))=1$, and the same happens for any $\theta\in\Theta$.
	An equivalent direct signaling scheme would publicly reveal a tuple of candidates' suggestions.
	For example, \textsf{not A} would become $(c_1,c_0,c_0)$.
\end{example}

We consider two commonly adopted voting rules: \emph{k-voting rule} and \emph{plurality voting rule}~(see, \emph{e.g.}, \cite{brandt2016handbook}).
%
%
In an election with a \emph{k-voting rule} each voter chooses a candidate after observing the sender's signal.
Candidate $c_i$ is elected if it receives at least $k$ votes, where $k\in [|R|]$ is the established electoral rule.~\footnote{We denote by $[n]$ the set $\{1,\ldots,n\}$.}
The problem of designing the optimal sender's persuasive signaling scheme under a $k$-voting rule is denoted by \textsf{K-V}.
In an election with \emph{plurality voting rule} the winner is determined as the candidate with a plurality (greatest number) of votes.
The problem of finding an optimal persuasive signaling scheme for the sender with plurality voting is denoted by \textsf{PL-V}.
In both settings, we focus on maximizing the winning probability of the sender.
The problem can be written as the optimization problem: $\max_{\phi}\sum_{\theta\in\Theta,\s\in S}\mu(\theta)\phi(\theta,\s)W(\s)$, subject to $\phi$ being persuasive for each voter.
An approximate solution for this problem is a signaling scheme that is persuasive, but guarantees to the sender a sub-optimal expected utility.
%
%

%% file: private_k.tex
\section{Private Signaling with $k$-Voting Rules}\label{sec:private_k}

In this section, we show that a solution to \textsf{K-V} (\emph{i.e.}, finding an optimal persuasive signaling scheme under a $k$-voting rule) can be found in polynomial time when the sender can employ a private signaling scheme.

First, we show that the sender can restrict the choice of a signaling scheme to the set of the schemes $\phi$ whose marginal signaling schemes are Pareto efficient on the set $\{\phi_r(\theta,c_0)\}_{\theta\in\Theta,r\in R}$ (Lemma~\ref{lemma:pareto}), and recommend with positive probability either $c_0$ or the candidate giving $r$ the best utility under $\theta$ (Lemma~\ref{lemma:positive_prob}).

\begin{lemma}\label{lemma:pareto}
	Given a signal $\phi'$ and a set of persuasive marginal signaling schemes $\{\phi_r\}_{r\in R}$, if $\phi_r(\theta,c_0)\geq\phi'_r(\theta,c_0)$ for each $r\in R$ and $\theta\in\Theta$, there exists a persuasive signaling scheme $\phi$ such that $W(\phi)\geq W(\phi')$.
\end{lemma}

\begin{proof}
	Be given a signaling scheme $\phi'$ and a set of persuasive marginal signaling schemes $\{\phi_r\}_{r\in R}$ s.t. $\phi_r(\theta,c_0)\geq \phi_r'(\theta,c_0)$ for each $r\in R$, $\theta\in \Theta$.
	Intuitively, we show that it is possible to move probability mass to $c_0$ while guaranteeing persuasiveness with the following iterative procedure.

	Let $\phi^0=\phi'$. 
	Then, we iterate over $r\in [|R|]$, and update the signaling scheme with the following procedure.
	Let $A_r$ be an arbitrary mapping from $[|S_{-r}|]$ to $S_{-r}$, which serves as an arbitrary ordering of elements in $S_{-r}$ (\emph{i.e.}, $A_r(i)$ returns the $i$-th element of $S_{-r}$ in such ordering).
	Moreover, for each $\theta\in \Theta$, we define $\Delta_r^{0}(\theta)=\phi_r(\theta,c_0) -\phi'_r(\theta,c_0)$.
	For each $r$, we iterate over $i\in[|S_{-r}|]$, and perform the following updates:
	$\s_{-r}=A_r(i)$, 
	\begin{multline}\label{eq:lemma1_1}\phi^r(\theta,(c_0,\s_{-r}))=\min\bigg\{\phi^{r-1}(\theta,(c_0,\s_{-r}))+\Delta_r^{i-1}(\theta),\\ \sum_{c \in C} \phi^{r-1}(\theta,(c_0,\s_{-r}))\bigg\},\end{multline}
	and
	\begin{multline*}
	\Delta_r^i(\theta)= \Delta_r^{i-1}(\theta) -\phi^r(\theta,(c_0,\s_{-r}))+ \phi^{r-1}(\theta,(c_0,\s_{-r})),
	\end{multline*}
	where $\phi^r(\theta,(c,\s_{-r}))$ is the probability of recommending $c$ to $r$ and $\s_{-r}$ to the other receivers, under $\theta$ (at iteration $r$).
	Finally, for each $\s_{-r}$, and $c \neq c_0$, set: 

	\begin{multline*}\label{eq:lemma1_2}\phi^r(\theta,(c,\s_{-r}))= \\= \frac{\phi_r(\theta,c)\bigg(\sum_{c' \in C}\phi^{r-1}(\theta(c',\s_{-r})) - \phi^r(\theta,(c_0,\s_{-r}))\bigg)}{\sum_{c' \in C \setminus \{c_0\}} \phi_r(\theta,c')},
	\end{multline*}
	the numerator is well-defined because of the minimization in Equation~\ref{eq:lemma1_1}.
	After having enumerated all the receivers, we obtain $\phi^{|R|}$.
	We show that $\phi=\phi^{|R|}$ is precisely the desired signaling scheme.
	%
	First, we show that, at each iteration $r$, $\phi^r$ is well formed.
	For each iteration $r$, and pair $(\theta,\s_{-r})$, we show that $\sum_{c \in C} \phi^{r}(\theta,(c,\s_{-r}))=\sum_{c \in C} \phi^{r-1}(\theta,(c,\s_{-r}))$.
	We have: $\sum_{c \in C} \phi^{r}(\theta,(c,\s_{-r}))= \phi^{r}(\theta,(c_0,\s_{-r}))+ \sum_{c \in C\setminus\{c_0\} } \phi^{r}(\theta,(c,\s_{-r}))$.
	Then, by expanding $\phi^{r}(\theta,(c,\s_{-r}))$ via the update rule, we obtain:
	\begin{multline*}
	\sum_{c \in C} \phi^{r}(\theta,(c,\s_{-r}))=\\= \phi^{r}(\theta,(c_0,\s_{-r})) + \sum_{c \in C } \phi^{r-1}(\theta,(c,\s_{-r})) - \phi^{r}(\theta,(c_0,\s_{-r})),
	\end{multline*}
	which is precisely $\sum_{c \in C} \phi^{r-1}(\theta,(c,\s_{-r}))$.
	This implies that $\sum_{\s \in S} \phi^r(\theta, \s)=1$, and that receiver $r$'s marginal probabilities are modified only at iteration $r$.
	Now, we show that receiver $r$'s marginals are updated correctly.
	We distinguish the following two cases. 
	
	i) It is easy to see that, for candidate $c_0$,
	\begin{multline*}
	\sum_{\s_{-r} \in S_{-r}}
	 \phi^r(\theta,(c_0,\s_{-r}))=\\
	 = \Delta_r^0(\theta) + \sum_{\s_{-r} \in S_{-r}} \phi^{r-1}(\theta,(c_0,\s_{-r}))= \phi_r(\theta,c_0).
	\end{multline*}

	ii) For each candidate $c\neq c_0$, we have:
	\begin{align*}
	& \sum_{\s_{-r} \in S_{-r}} \phi^r(\theta,(c,\s_{-r})) =\\
	& = \hspace{-.3cm}\sum_{\s_{-r} \in S_{-r}}\hspace{-.3cm}  \frac{\phi_r(\theta,c) \bigg(\displaystyle\sum_{c' \in C}\phi^{r-1}(\theta(c,\s_{-r})) - \phi^r(\theta,(c_0,\s_{-r}))\bigg)}{\displaystyle\sum_{c' \in C\setminus\{c_0\}} \phi_r(\theta,c')} =\\
	& = \frac{\phi_r(\theta, c)\bigg( \displaystyle\sum_{\s \in S} \phi^{r-1}(\theta,\s) - \sum_{\s_{-r} \in S_{-r}}  \phi^r(\theta,(c_0,\s_{-r}))\bigg)}{\displaystyle\sum_{c' \in C\setminus\{c_0\}} \phi_r(\theta, c')}=\\
	& = \frac{\phi_r(\theta,c)(1-\phi_r(\theta, c_0))}{\displaystyle\sum_{c' \in C\setminus\{c_0\}} \phi_r(\theta,c')}=\phi_r(\theta,c).
	\end{align*}
	Since $\{\phi_r\}_{r\in R}$ are persuasive, also the new signaling scheme $\phi$ is persuasive.
	Finally, we show that the new signaling scheme does not decrease sender's expected utility. 
	Let $S^\ast=\{\s\in S|\delta(\s,c_0)\geq k\}$ be the set of joint signals recommending $c_0$ to more than $k$ voters (under a $k$-voting rule).
	Then, $W(\phi)=\sum_{\theta\in\Theta}\mu(\theta) \sum_{\s\in S^\ast}\phi(\theta,\s)$.
	It is enough to show that, for each iteration $r$, for each $\theta\in \Theta$, and, for each $\s_{-r}\in S_{-r}$, it holds
	\[
	\sum_{c\in C}\left(\phi^r(\theta,(c,\s_{-r}))-\phi^{r-1}(\theta,(c,\s_{-r}))\right) \mathds{1}_{(c,\s_{-r})\in S^\ast}\geq 0.
	\] 
	We distinguish three cases. 
	i) When $\delta(\s_{-r},c_0)<k-1$, a change in $r$'s marginal probabilities does not affect the sender's winning probability,  term $\mathds{1}_{(c,\s_{-r})\in S^\ast}$ being always $0$. 
	ii) When $\delta(\s_{-r},c_0)= k-1$, $\mathds{1}_{(c,\s_{-r})\in S^\ast}=1$ only if $c=c_0$, and
	$\phi^r(\theta,(c_0,\s_{-r}))\ge \phi^{r-1}(\theta,(c_0,\s_{-r}))$.
	iii) When $\delta(\s_{-r},c_0)>k-1$, $\mathds{1}_{(c,\s_{-r})\in S^\ast}$ is always $1$, and we already know that $\sum_{c\in C}\left(\phi^r(\theta,(c,\s_{-r})) - \phi^{r-1}(\theta,(c,\s_{-r}))\right)=0$. 
	This concludes the proof.
\end{proof}

We now state the next lemma.

\begin{lemma}\label{lemma:positive_prob}
	There always exists a solution to \textsf{K-V} in which, for all $r\in R$ and $\theta\in \Theta$, $\phi_r(\theta,c)>0$ if and only if one of the following two conditions is satisfied: 
	\begin{itemize}[topsep=-1mm,itemsep=-1mm]
		\item $c=c_0$,
		\item $c\in\argmax_{c'\in C} u_r(\theta,c')$.
	\end{itemize}
\end{lemma}
\begin{proof}
	Given a persuasive signaling scheme $\phi'$, we show that it is possible to build a collection $\{\phi_r\}_{r\in R}$ with the property above, s.t. $\phi_r(\theta,c_0)\geq \phi'_r(\theta,c_0)$ for each $r\in R$, $\theta\in \Theta$.
	This, together with Lemma~\ref{lemma:pareto}, proves our result.
	We build $\phi$ iteratively. 
	For each pair $(\theta,r)$, select $c^\ast\in\argmax_{c\in C}u_r(\theta, c)$, and set $\phi_r(\theta,c^\ast)=1-\phi'_r(\theta, c_0)$, $\phi_r(\theta,c_0)=\phi_r'(\theta,c_0)$, and $\phi_r(\theta,c)=0$ for each other $c\in C\setminus \{c_0,c^\ast\}$.
	It is immediate to see that, for each $\theta$ and $r$, $\sum_{c\in C}\phi_r(\theta,c)=1$.
	Next, we show that each $\phi_r$ is persuasive, i.e., $\sum_{\theta\in\Theta}\mu(\theta)\phi_r(\theta,c)\left(u_r(\theta,c)-u_r(\theta,c')\right)\geq 0$ for each $r\in R$, and $c,c'\in C$.
	If $c=c_0$, we have $\phi_r(\theta,c_0)>\phi'_r(\theta,c_0)$ only if $c_0\in \argmax_{c\in C} u_r(\theta,c)$, which means $\left(u_r(\theta,c_0)-u_r(\theta,c')\right)\geq 0$, in the remaining cases we have $\phi_r(\theta,c_0)=\phi_r'(\theta,c_0)$.
	If $c\neq c_0$, $c\in\argmax_{c'\in C} u_r(\theta,c')$ for each $\theta\in\Theta$ with $\phi(\theta,c)>0$, which makes the incentive constraint satisfied.
\end{proof}


By exploiting Lemma~\ref{lemma:positive_prob}, we show that an optimal persuasive signaling scheme under a $k$-voting rule can be computed in polynomial time via the following linear program (LP).
Let $\beta_\theta\in\mathbb{R}$ be the probability with which $k$ voters vote for $c_0$ with state $\theta$.
Then, we can compute an optimal solution to \textsf{K-V} as follows (the proof is provided below):
\begin{subequations}\label{eq:lp1}
	\begin{align}
		\max_{\substack{\beta\in [0,1]^{|\Theta|},z\in\mathbb{R}_{-}^{|\Theta|\times k\times |R|}\\
				t,q\in\mathbb{R}^{|\Theta|\times k}\\ \phi_r(\cdot,c_0)\in[0,1]^{|R\times \Theta|}\\ }} & \sum_{\theta\in\Theta}\mu(\theta)\beta_\theta \label{eq:lp1_obj}\\ & 	\hspace{-2.6cm}\textnormal{s.t.}\sum_{\theta\in\Theta}\mu(\theta)\phi_r(\theta,c_0)\left(u_r(\theta,c_0)-u_r(\theta,c)\right)\geq 0 \label{eq:lp1_persuasive}\\
		\nonumber & \hspace{1.15cm}\forall r\in R,\forall c\in C\setminus \{c_0\} \\
		& \hspace{-2.4cm}\beta_\theta\leq\frac{1}{k-m} q_{\theta,m} \hspace{0.05cm}\forall\theta\in\Theta, \forall m\in\{0,\ldots,k-1\}\label{eq:lp1_beta_ub}\\
		& \hspace{-2.4cm}q_{\theta,m}\leq (|R|-m)t_{\theta,m}+\sum_{r\in R}z_{\theta,r,m}\label{eq:lp1_q_ub}\\
		\nonumber & \hspace{0.35cm}\forall \theta\in\Theta,\forall m \in \{0,\ldots,k-1\}\\
		& \hspace{-2.4cm}\phi_r(\theta,c_0)\geq t_{\theta,m} + z_{\theta,m,r} \label{eq:lp1_phi_lb}\\
		& \nonumber \hspace{-0.85cm}\forall r\in R,\forall \theta\in\Theta, \forall m\in \{0,\ldots,k-1\}.
	\end{align}
\end{subequations}  
\

This formulation allows us to state the following:
\begin{theorem}\label{th:poly-k-v}
	It is possible to compute an optimal persuasive private signaling scheme for \textsf{K-V} in \emph{poly($n$, $\ell$, $|R|$)} time. 
\end{theorem}
\begin{proof}
	Formulation~\ref{eq:lp1} has a polynomial number of variables and constraints.
	Then, proving Theorem~\ref{th:poly-k-v} amounts to show that a solution to Formulation~\ref{eq:lp1} is also a solution to \textsf{K-V}.
	
	Let $c_{\theta,r}^\ast=\argmax_{c\in C}u_r(\theta,c)$, for each $\theta$ and $r$.
	First, by Lemma~\ref{lemma:positive_prob}, the space of available signals can be restricted to those in which, for each $r$ and $\theta$, only $\phi_r(\theta,c_0)$ and $\phi_r(\theta,c^\ast_{\theta,r})$ are $>0$, and $\phi_r(\theta,c^\ast_{\theta,r})=1-\phi_r(\theta,c_0)$.
	Constraints~\eqref{eq:lp1_persuasive} are the incentive constraints for action $c_0$.
	For any $c\neq c_0$, the incentive constraints are satisfied by construction.
	Objective~\eqref{eq:lp1_obj} is given by the sum over all $\theta\in\Theta$ of the prior of state $\theta$, multiplied by the probability of having at least $k$ vote for $c_0$ given $\theta$.
	We need to show the correctness of $\beta_\theta$.
	For each state of nature the maximum probability with which at least $k$ receivers play $c_0$ is given by:
	\[
	\beta_\theta=\min\left\{\min_{m\in\{0,\ldots,k-1\}} \frac{1}{k-m}q_{\theta,m}, 1\right\},
	\]
	where $q_{\theta,m}$ is the sum of the lowest $|R|-m$ elements in the set $\{\phi_r(\theta,c_0)\}_{r\in R}$; for further details, see~\citep[Lemma 3]{arieli2019private}.
	This is enforced via Constraints~\eqref{eq:lp1_beta_ub}.
	Constraints~\eqref{eq:lp1_q_ub} and~\eqref{eq:lp1_phi_lb} ensure $q_{\theta,m}$'s consistency, and are derived from the dual of a simple LP of this kind: $\min_{y\in\mathbb{R}^n}x^\top y$ s.t. $\mathds{1}^\top y=w$ and $0\leq y\leq 1$ (where $x\in\mathbb{R}^n$ is the vector from which we want to extract the sum of the smallest $w$ entries).
	This concludes the proof.
\end{proof}

%% file: private_general.tex
\section{A Condition for Efficient Private Signaling}

%
%

In the following, we provide a necessary and sufficient condition for the poly-time computation of persuasive private signaling schemes under a general class of sender's objective functions. 
In the next section, this result will be exploited when dealing with anonymous utility functions and plurality voting.
We allow for general sender's utility functions of type $f_\theta:\times_{r\in R} C\to \mathbb{R}$, which generalizes previous results by~\cite{dughmi2017algorithmic} where the receivers' action space has to be binary. 
Given a collection of set functions $\mathcal{F}$, $P(\mathcal{F})$ denotes the class of persuasion instances in which, for each $\theta\in\Theta$, $f_\theta\in\mathcal{F}$.
We can state the following (the proof can be found in the Supplementary Material).

\begin{restatable}{theorem}{general}\label{th:general}
%
	Let $\mathcal{F}$  be any collection of set functions including $f_0(\cdot)=0$.
	Given any instance in $P(\mathcal{F})$, there exists a polynomial-time algorithm for computing an optimal persuasive private signaling scheme if and only if there is a polynomial-time algorithm that computes 
	\begin{equation}\label{eq:general}
	\max_{\ct\in\times_{r\in R}C}f(\ct)+\sum_{r\in R}w_r(c^r),
	\end{equation}
	for any $f\in\mathcal{F}$, and any weights $w_r(c^r)\in\mathbb{R}$, where $c^r$ is the action chosen by $r$ in $\ct$.
\end{restatable}

The crucial difference with the result by~\cite{dughmi2017algorithmic} is that they consider set functions 
depending only on the set of players choosing the target action, among the two available.
Theorem~\ref{th:general} generalizes this setting as it allows for functions taking as input any action profile $\ct$.
This is crucial in settings like plurality voting, where the sender is not only interested in votes favorable to $c_0$, but also in the distribution of the other preferences.
\cite{dughmi2017algorithmic}'s result cannot be applied to such settings.

%% file: private_plurality.tex
\section{Further Positive Results for Private Signaling}

Despite Theorem~\ref{th:general}, in the case of general utility functions the problem of determining an optimal persuasive private signaling scheme is still largely intractable.
An intuition behind that is that there may be an exponential (in $|C|$) number of values of $f$ (\emph{e.g.}, in the case of anonymous utility functions, there are $\binom{|R|+|C|-1}{|R|}$ values of $f$).
In order to identify tractable classes of the problem, we need to make some further assumptions on $\mathcal{F}$.

\textbf{Anonymous Utility Functions}.
A reasonable (in the context of voting) restriction is to \emph{anonymous utility functions} (see, \emph{e.g.}, \cite{arieli2019private}).
Previous results on the computational complexity of private signaling with anonymous utility functions focus on the case of binary actions, which is shown to be tractable~\citep{babichenko2016computational,arieli2019private,dughmi2017algorithmic}.
We generalize these results to a generic number of states of nature and receiver's actions with the following result (the proof is provided in the Supplementary Material).
\begin{restatable}{theorem}{anon}\label{th:anonymous}
	Private Bayesian persuasion with anonymous sender's utility functions is fixed-parameter tractable in the number of receivers' actions.
\end{restatable}

Theorem~\ref{th:anonymous} implies that, for any anonymous voting rule, private Bayesian persuasion is fixed-parameter tractable in the number of candidates.

\textbf{Plurality Voting}.
By further restricting our attention to specific voting rules, we can see the consequences of Theorem~\ref{th:general} to an even better extent.
A simple and widespread voting rule is \emph{plurality voting}.~\footnote{See, \emph{e.g.}, its (\emph{discussed}) adoption in direct presidential elections in a number of states~\citep{blais1997direct}.}
In this setting $W(\s)=1$ if and only if $\delta(\s,c_0)>\delta(\s,c)$ for any $c\neq c_0$, and $W(\s)=0$ otherwise.
We can state the following:
\begin{theorem}\label{th:private_pl}
\textsf{PL-V} with private signaling can be solved in $\textnormal{poly}(n,\ell,|R|)$ time.	
\end{theorem}
\begin{proof}
	We exploit Theorem~\ref{th:general}, and show that the maximization Problem~\eqref{eq:general} can be solved efficiently.
	With an overload of notation, generic actions profiles are represented via signals.
	Then, the maximization problem reads: $\max_{\s \in S} W(\s) +  \sum_{r\in R}w_r(s_r)$.
	We split the maximization problem in two steps.
	First, we consider the maximization over non-winning action profiles, \emph{i.e.}, signals in the set $\bar S=\{\s \in S | \exists c \neq c_0 \textnormal{ s.t. } \delta(\s,c)>\delta(\s,c_0)\}$.
	An upper bound to the optimal value of the maximization problem restricted to $\bar S$ is given by $\max_{\s \in S}\sum_{r}w_r(s_r)$.
	The latter problem can be solved independently for each receiver $r$, by choosing $c$ maximizing $w_r(c)$.
	Once the relaxed problem has been solved, the objective function of the separation problem is adjusted by checking whether $\ct$ is winning or not.
	The resulting value is then compared with the value from the following step.
	
	We consider the maximization over winning action profiles, \emph{i.e.}, signals in $S^\ast=S \setminus \bar S$.
	For any $\s\in S^\ast$, $W(\s)=1$.
	Then, we have to maximize the same objective of the previous case with the following additional constraints: $\delta(\s,c_0)>\delta(\s,c)$, for all $c\neq c_0$.
	To determine an optimal solution to this problem, we enumerate over $k\in\{\lceil\frac{|R|-1}{\ell}\rceil+1,\ldots,|R|\}$, \emph{i.e.}, the number of votes that make $c_0$ a potential winner of the election.
	Then, for each value of $k$, we consider action profiles such that $\delta(\s,c_0)=k$, and $\delta(\s,c)<k$, for all $c\neq c_0$ (\emph{i.e.}, winning signals where $c_0$ receives exactly $k$ votes).
	An optimal solution for a fixed $k$ can be determined with this LP:
	\vspace{.1cm}
	\begin{subequations}
		\begin{align*}
		\max_{\chi\in\mathbb{R}_{+}^{|R\times C|} }& \sum_{(r,c)\in R\times C}\chi_r(c)w_r(c)\\
		\textnormal{s.t. } & \sum_{r\in R}\chi_r(c_0)=k\\
		& \sum_{r\in R} \chi_r(c)\leq k-1 \hspace{.5cm}\forall c\in C\setminus\{c_0\}\\
		& \sum_{c\in C}\chi_r(c)=1 \hspace{1.1cm}\forall r\in R.
		\end{align*}
	\end{subequations}
	We look for an integer solution of the problem, which always exists and can be found in polynomial time (see, \emph{e.g.},~\citep{orlin1997polynomial}).
	This is because the formulation is an instance of the \emph{maximum cost flow problem}, which is, in its turn, a variation of the \emph{minimum cost flow problem}.
	%
	%
	Once an integer solution has been found, an optimal action profile of the original maximization problem is the one obtained by recommending to each $r$ the candidate $c$ s.t. $\chi_r(c)=1$.
\end{proof}

%% file: public.tex
\section{Public Signaling}\label{sec:public}

In contrast with the results for private signaling problems, we show that public persuasion in the context of voting is largely intractable.

We reduce from MAXIMUM $k$-SUBSET INTERSECTION (MSI)~\citep{clifford2011maximum}.
\begin{definition}[MSI]
	An instance of MAXIMUM $k$-SUBSET INTERSECTION is a tuple $(\mathcal{E}, A_1,\ldots,A_m,k,q)$, where $\mathcal{E}=\{e_1,\ldots,e_n\}$ is a finite set of elements, each $A_i$, $i\in [m]$, is a subset of $\mathcal{E}$, and $k$, $q$ are positive integers. 
	It is a  ``\emph{yes}''-instance if there exist exactly $k$ sets $A_{i_1},\ldots,A_{i_k}$ such that $|\cap_{j\in [k]} A_{i_j}|\geq q$, and a ``\emph{no}''-instance otherwise.
\end{definition}
MSI has been recently shown to be \textsf{NP}-hard~\citep{xavier2012note,elkind2015equilibria}.
Now, we prove the following negative result:
\begin{theorem}\label{th:negative}
	\textsf{K-V} with public signaling, even with two candidates, cannot be approximated in polynomial time to within \emph{any} factor, unless \textsf{P=NP}.
\end{theorem}
\begin{proof}
	Given an instance of MSI, we build a public signaling problem with the following features.

	\textbf{Mapping}. It has a voter $r_i$ for each $A_i$, $i\in[m]$, and $m$ voters $r_{e,j}$, $j\in [m]$, for each $e\in\mathcal{E}$.
	There is a state of nature $\theta_e$ for each $e\in\mathcal{E}$, and $\mu(\theta_e)=1/n$ for each $\theta_e$.
	Finally, $C=\{c_0,c_1\}$.
	Receivers have the following utility functions: for each $r_i$, $i\in[m]$,
	\begin{equation*}
	u_{r_i}(\theta_e,c) = \begin{cases}
	\begin{array}{ll}
	1 &\textnormal{if}\quad e\in A_i \textnormal{, } c=c_0\\
	-n^2 & \textnormal{if}\quad e\notin A_i \textnormal{, } c=c_0\\
	0 & \textnormal{if} \quad c=c_1\\
	\end{array}
	\end{cases},
	\end{equation*}
	for each $r_{e,j}$, $e\in\mathcal{E}$, and $j\in[m]$, 
	\begin{equation*}
	u_{r_{e,j}}(\theta_{e'},c) = \begin{cases}
	\begin{array}{ll}
	1 &\textnormal{if}\quad e=e'\textnormal{, }c=c_0\\
	-\frac{1}{q-1} & \textnormal{if}\quad e\neq e'\textnormal{, }c=c_0\\
	0 & \textnormal{if} \quad c=c_1\\
	\end{array}
	\end{cases}.
	\end{equation*}
	The sender needs at least $k+mq$ votes (for $c_0$) in order to win the election (i.e., we are considering a $(k+mq)$-voting rule).
	We prove our theorem by showing that $c_0$ has a strictly positive probability of winning the election if and only if the corresponding MSI instance is satisfiable.
	
	\textbf{If}.
	Suppose there exists a set $A^\ast=\{A_{i_1},\ldots,A_{i_k}\}$ satisfying the MSI instance, and let $I=\cap_{j\in[k]}A_{i_j}$.
	Define a signaling scheme $\phi$ with two signals ($\gamma_0$ and $\gamma_1$) such that, for each $e\in I$, $\phi(\theta_e,\gamma_0)=1$, and, for each $e\notin I$, $\phi(\theta_e,\gamma_1)=1$, and it is equal to 0 otherwise.
	We show that such a signaling scheme guarantees a strictly positive winning probability for the sender.
	First, we show that, when the realized state of nature $\theta_e$ is such that $e\in I$ (\emph{i.e.}, the sender recommends $\gamma_0$), at least $k+mq$ receivers vote for $c_0$. 
	Each receiver $r_i$ such that $A_i\in A^\ast$ will choose $c_0$ when recommended $\gamma_0$.
	Specifically, $\sum_{\theta_e}\frac{1}{n}\phi(\theta_e,\gamma_0)u_{r_i}(\theta_e,c_0)=\frac{q}{n}$, while $\sum_{\theta_e}\frac{1}{n}\phi(\theta_e,\gamma_0)u_{r_i}(\theta_e,c_1)=0$.
	Receivers $r_{e,j}$ with $e\in I$ will vote for $c_0$ after observing $\gamma_0$.
	This is because, for each $e\in I$ and $j\in[m]$, $r_{e,j}$ has expected utility $\frac{1}{n}\phi(\theta_e,\gamma_0)-\sum_{\theta_e':e'\neq e}\frac{1}{n}\frac{1}{q-1}\phi(\theta_{e'},c_0)=0$ for voting $c_0$, and expected utility 0 for voting $c_1$.
	Then, when the realized state of nature is $\theta_e$ with $e\in I$, there are at least $k+mq$ receivers voting for $c_0$.
	Therefore, the sender's winning probability is at least $\frac{k}{n}$ (\emph{i.e.}, the probability of observing $\theta_e$ with $e\in I$ under a uniform prior).
	
	\textbf{Only if}.
	Suppose, by contradiction, that MSI is not satisfiable, and that the sender's winning probability under the optimal signaling scheme is not null. 
	This implies the existence of a signal $\gamma_0$ such that, when recommended, a set of receivers $R^\ast$ votes for $c_0$, and $|R^\ast|\geq k+mq$.
	Then, there exist at least $q$ states $\theta_e$ in which all voters $r_{e,j}$, $j\in[m]$, vote for $c_0$.
	Each receiver $r_{e,j}$, having observed $\gamma_0$, votes for $c_0$ only if 
	$ \phi(\theta_e,\gamma_0)-\frac{1}{q-1}\sum_{\theta_{e'}:e'\neq e}\phi(\theta_{e'},\gamma_0)\geq 0$.
	This implies that  $\phi(\theta_e,\gamma_0)-  \sum_{\theta_{e'}} \phi(\theta_{e'},\gamma_0) + \phi(\theta_e,\gamma_0) \geq 0$ and $\phi(\theta_e,\gamma_0) \ge \sum_{\theta_{e'}\in \Theta}\phi(\theta_{e'},\gamma_0)/q$.
	Then, there are exactly $q$ states $\theta_e$ in which $\gamma_0$ is played with probability $\sum_{\theta_{e'}\in \Theta}\phi(\theta_{e'},\gamma_0)/q$, while $\gamma_0$ is never played in the remaining states.
	As a consequence, $R^*$ includes exactly $mq$ voters $r_{e,j}$, and at least $k$ voters $r_i$. 

	Each voter $r_i\in R^\ast$, after observing $\gamma_0$, choose candidate $c_0$.
	Therefore, $\sum_{\theta_e\in\Theta}\mu(\theta_e)\phi(\theta_e,\gamma_0)(u_{r_i}(\theta_e,c_0)-u_{r_i}(\theta_e,c_1))\geq 0$.
	We obtain $\sum_{e\in A_i}\phi(\theta_e,\gamma_0)-n^2\sum_{e\notin A_i}\phi(\theta_e,\gamma_0)\geq 0$. 
	Then,
	\[
	\sum_{e\in A_i}\phi(\theta_e,\gamma_0)-n^2\sum_{e\in \mathcal{E}}\phi(\theta_e,\gamma_0)+n^2\sum_{e\in A_i}\phi(\theta_e,\gamma_0)\geq 0.
	\]
	Let $\xi(A_i)=\sum_{e\in A_i}\phi(\theta_e,\gamma_0)$. 
	We have $\xi(A_i)\geq \frac{n^2}{n^2+1}\sum_{e\in\mathcal{E}}\phi(\theta_e,\gamma_0)$ for each $i\in[m]$ such that $r_i\in S^\ast$.

	
	
	Let $\mathcal{E}^\ast$ be the set of elements $e$ such that $r_{e,j}\in S^\ast$, for all $j\in[m]$.
	%
	%
	In this case, since MSI is not satisfiable, there exists a pair $(r_{\hat i},e)\in R\times\mathcal{E}^\ast$ such that $r_{\hat i}\in R^\ast$ and $e\notin A_{\hat i}$ (otherwise $\{A_i\}_{i:r_i\in R^\ast}$ would be a feasible solution with intersection $\mathcal{E}^\ast$).
	We observed that, in each $\theta_e$ with $e\in\mathcal{E}^\ast$, $\gamma_0$ is recommended with probability $\sum_{e\in\mathcal{E}}\phi(\theta_e,\gamma_0)/q$.
	Then, $\xi(A_{\hat i})=\sum_{e\in A_{\hat i}}\phi(\theta_e,\gamma_0)\leq \frac{q-1}{q}\sum_{e\in\mathcal{E}}\phi(\theta_e,\gamma_0)$.
	This leads to a contradiction since \[\frac{q-1}{q}\sum_{e\in\mathcal{E}}\phi(\theta_e,\gamma_0)\geq \frac{n^2}{n^2+1}\sum_{e\in\mathcal{E}}\phi(\theta_e,\gamma_0)\] has no solutions (since $q$ and $n$ are positive integers and $q\leq n$).
	This concludes our proof.
\end{proof}

Theorem~\ref{th:negative} improves the negative results provided by~\cite{dughmi2017algorithmic} (Theorem 6.2), where they show that optimal sender's utility cannot be approximated to within any constant multiplicative factor, unless $\mathsf{P=NP}$.
Our result strengthen the negative result by~\cite{dughmi2017algorithmic} by extending the inapproximability to \emph{any} factor that is function of the input size, thus even excluding approximation factors decreasing as the input size increases.

Moreover, Theorem~\ref{th:negative} implies that the public signaling problem is intractable even with more general sender's utility functions.
It is immediate to see that the same negative result holds for anonymous utility functions (a $k$-voting rule induces a sender's anonymous utility function), and we prove that the same hardness result holds also for plurality voting (see the Supplementary Material).
\begin{restatable}{corollary}{corollaryNeg}
	\textsf{PL-V} with public signaling, even with two candidates, cannot be approximated in polynomial time to within \emph{any} factor, unless \textsf{P=NP}.
\end{restatable}

%% file: conclusion.tex
\section{Discussion and Future Research}

This paper studies how a malicious actor may influence the outcome of a voting process by the strategic provision of information to voters that update their beliefs rationally.
We focus on the case with no inter-agent externalities, and allow for an arbitrary number of candindates and states of nature, thus generalizing the model of~\cite{arieli2019private}.
We draw a sharp contrast between the tractability of the problem of computing an optimal persuasive signaling scheme in the private and public case, respectively.
In the former setting, we show that under various voting rules (\emph{i.e.}, $k$-voting rules and plurality voting) the problem can be solved efficiently.
In doing so, we provide a generalization of the fundamental necessary and sufficient condition first described by~\cite{dughmi2017algorithmic}.
In the public signaling case, we propose a new inapproximability result which strongly improves previously known results, showing that in this setting the problem is unlikely to be tractable.

In the future, we are interested in identifying special classes of instances in which optimal persuasive public signaling schemes can be found efficiently and in using our results for the analysis of other voting rules.
Moreover, we are interested in the extremely challenging case in which multiple competing senders face the problem of manipulating the same election.

%% file: appendix.tex
\appendix
\section{Supplementary Material}
\bigskip

\section{Omitted Proofs}

\general*

\begin{proof}
	Given a set $\{f_\theta\}_{\theta\in\Theta}$, the persuasion problem can be formulated with the following LP: 
	\begin{subequations}\label{eq:lp_general}
		\begin{align}
		\max_{x\in[0,1]^{|\Theta\times S|}} & \sum_{\theta\in\Theta,\s\in S}x(\theta,\s) f_\theta(\s)\label{eq:lp_general_obj}\\
		&\sum_{\substack{\theta\in\Theta,\\\s:s_r=c}} x(\theta,\s)\left(u_r(\theta,c)-u_r(\theta,c')\right)\geq 0\label{eq:lp_general_ic}\\
		&\nonumber \hspace{3cm}\forall r\in R, \forall c,c'\in C\\
		& \sum_{\s\in S} x(\theta,\s) = \mu(\theta) \hspace{1.9cm} \forall \theta\in\Theta \label{eq:lp_general_phi}
		\end{align}
	\end{subequations}
	Note that constraints~\ref{eq:lp_general_ic} force the signaling scheme to be persuasive.
	Therefore, in objective~\ref{eq:lp_general_obj}, we can write $f_\theta(\s)$ in place of $f_\theta(\ct)$.
	
	$(\Longrightarrow)$.
	Let $y\in\mathbb{R}_{-}^{|R\times C\times C|}$ be the dual variables of primal constraints~\ref{eq:lp_general_ic} and $d\in\mathbb{R}^{|\Theta|}$ be the dual variables of constraints~\ref{eq:lp_general_phi}.
	The dual of LP~\ref{eq:lp_general} has a polynomial number of variables and an exponential number of constraints, one for  each pair $(\theta,\s)\in \Theta\times S$, of type:
	\begin{multline*}
	O(\theta,\s)=\left( - \sum_{\substack{r\in R,\\c\in C}}y_r(s_r,c)\left(u_r(\theta,s_r)- u_r(\theta,c)\right)\right) + \\ -d(\theta)  + f_\theta(\s)\leq 0.
	\end{multline*}
	We show that, given a vector of dual variables $\bar z=(\bar y,\bar d)$, the problem of either finding a hyperplane separating $\bar z$ from the set of feasible solutions to the dual or proving that no such hyperplane exists can be solved in polynomial time. 
	The \emph{separation problem} of finding an inequality of the dual which is maximally violated at $\bar z$ reads: $\max_{(\theta,\s)\in\Theta\times S} O(\theta,\s)$.
	A pair $(\theta,\s)$ yielding a violated inequality exists if and only if the separation problem admits an optimal solution of value $>0$.
	One such pair (if any) can by found in polynomial time by enumerating over states in $\Theta$.
	For each $\theta$, the problem reduces to $\max_{\s} \sum_{r\in R}v_r(\theta,s_r)+f_\theta(\s)$, where $v_r(\theta,s_r)=-\sum_{c\in C}\bar y_r(s_r,c) \left(u_r(\theta,s_r)-u_r(\theta,c)\right)$.
	It is enough to take $w_r(c)=v_r(\theta,c)$ to complete the \emph{if} part of the proof.	
	
	$(\Longleftarrow)$.
	Given a poly-time algorithm to determine an optimal signaling scheme for any instance of $P(\mathcal{F})$, we want to show that $\max_{\s\in S} f(s)-\sum_{r\in R} w_r(s_r)$ can be solved efficiently for any $\{w_r(c)\}_{r,c}$, and $f\in\mathcal{F}$.
	
	To reduce this problem to a signaling problem we employ a duality-based analysis introduced in~\cite{dughmi2017algorithmic}, and later improved by~\cite{xu2019tractability}.
	Our generalization to non-binary action spaces requires a more involved proof, as we will highlight in the following.
	Moreover, our proof completely diverges from ~\cite{dughmi2017algorithmic}’s and ~\cite{xu2019tractability}’s in the final construction of the mapping to a private signaling problem.

	%
	
	Given a set of weights $\{\bar w_r(c)\}_{r,c}$, and $f\in\mathcal{F}$, we are interested in the maximization of $\bar f(\s)=f(\s)+\sum_{r\in R}\bar w_r(s_r)$ over $S$.
	First, we slightly modify weights by setting, for each $r\in R$, $\bar w_r(c)\leftarrow \bar w_r(c) - \max_{c'}\bar w_{r}(c')$, for each $c\in C$.
	This modification preserves the set of optimal solutions of the maximization problem.
	After that, for each receiver $r$, it holds $\bar w_r\leq 0$, and there exists $\hat c^r\in C$ s.t. $\bar w_r(\hat c^r)=0$.
	Let, for each $r\in R$, $C_r=C\setminus \{\hat c^r\}$ ($\hat c^r$ can be selected arbitrarily from the actions s.t. $\bar w_r(c)=0$).
	We show that $\max_{\s\in S} \bar f(\s)$ can be reduced to solving the following LP, for all possible linear coefficients $\alpha$, $\{\beta_r(c)\}_{r\in R,c\in C_{r}}$:
	\begin{subequations}\label{eq:subproblem}
		\begin{align}
		\min_{\substack{z\in\mathbb{R}^{|R\times C_r|}\\ v\in\mathbb{R}}} & \sum_{r\in R, c\in C_{r}} \beta_r(c) z_r(c) + \alpha v\\
		\textnormal{s.t.} \hspace{.3cm}& \sum_{\substack{r\in R\\ s_r\neq \hat c^r}}z_r(s_r)+v\geq f(\s) \hspace{1cm}\forall \s\in S\\
		&z_r(c)\geq 0 \hspace{2cm}\forall r\in R,c\in C_{r}.\label{eq:subproblem_geq}
		\end{align}
	\end{subequations}
	To show this, we first argue that the maximization problem can be reduced to the separation problem for the feasible region of LP~\ref{eq:subproblem}.
	Take $z_r(c)=-\bar w_r(c)$ for all $r$ and $c\in C_r$.
	Constraints of family~\ref{eq:subproblem_geq} are satisfied by construction.
	Then, a pair $(\{\bar w_r(c)\}_{r,c}, v)$ is feasible if and only if $ v\geq \max_{\s} f(\s) + \sum_{r,s_r\neq\hat c^r} \bar w_r(s_r)$.
	As a result, the optimal value $v^\ast$ (which is the exact optimal objective of $\bar f(\s)$) can be determined via binary search in $O(B)$ steps, where $B$ is the bit complexity of the $f(\s)$'s and $w$'s.
	Then, by setting $\bar v=v^\ast-2^{-B}$, we obtain an infeasible pair $(\bar w,\bar v)$.
	If the separation oracle is given in input $(\bar w,\bar v)$, it returns a separating hyperplane corresponding to the optimal solution of the maximization problem.
	The equivalence between optimization and separation implies that the maximization problem reduces to solving LP~\ref{eq:subproblem} for any linear coefficients $\{\beta_r(c)\}_{r\in R,c\in C_r}$ and $\alpha$~\citep{khachiyan1980,grotschel1981ellipsoid}.
	
	A crucial difference between LP~\ref{eq:subproblem} and~\cite{xu2019tractability}'s analogous LP is that we modify the initial weights $\bar w$ to make them $\leq 0$ (simplifying the LP's structure), and, for each $r$, there is at least one $\bar w_r(c)$ equal to $0$.
	This reduces the number of variables in LP~\ref{eq:subproblem}, as variables $z_r(\hat c^r)$ are not included.
	This is fundamental for the last step of the proof.
	
	
	The next step is showing that LP~\ref{eq:subproblem} can be solved \emph{directly} for some parameters' values.
	Specifically:
	\begin{itemize}[topsep=0mm,itemsep=0mm]
		\item If $\alpha<0$ the solution is unbounded (\emph{i.e.}, the objective function tends to $-\infty$ as $v\to \infty$).
		\item If $\alpha=0$ and there exists $(\bar r,\bar c)$ s.t. $\beta_{\bar r}(\bar c)<0$, then a feasible solution is obtained by setting: $z_{\bar r}(\bar c)=v$, and  $z_r(c)=0$ for all $(r,c)\neq(\bar r,\bar c)$. 
		Again, for $v\to\infty$ the objective tends to $-\infty$.
		\item If $\alpha=0$ and $\beta_r(c)\ge 0$ for all $(r,c)$, then the objective is $\geq 0$ for any feasible solution. 
		By selecting a sufficiently large $v$ we obtain a feasible and optimal solution with objective value 0.
	\end{itemize}
	Therefore, when $\alpha\leq 0$ the problem can be solved in polynomial time.
	
	We focus on the case in which $\alpha>0$.
	Since $\alpha>0$, we can re-scale all coefficients of LP~\ref{eq:subproblem} by a factor $1/\alpha$ without affecting its optimal solutions, and obtain an equivalent LP with $\alpha=1$.
	The dual of LP~\ref{eq:subproblem} with $\alpha=1$ is:
	\begin{subequations}\label{eq:subproblem_dual}
		\begin{align}
		\max_{p\in\mathbb{R}_{+}^{|S|}} \hspace{.2cm}& \sum_{\s\in S} p(\s)f(\s)\\
		\textnormal{s.t.} & \sum_{\s: s_r=c} p(\s) \leq \beta_r(c) \hspace{.5cm}\forall r\in R, c\in C_{r}\label{eq:constr_dual_modificato}\\
		& \sum_{\s \in S} p(\s)=1
		\end{align}
	\end{subequations}
	Finally, we show that finding an optimal solution to LP~\ref{eq:subproblem_dual} reduces to finding an optimal signaling scheme in an instance of private persuasion with $|\Theta|=|C|$ states of nature, and $\mu(\theta)=\frac{1}{|C|}$ for each $\theta$. 
	First, for each $r$ we define an arbitrary one-to-one correspondence between elements of $C_r$, and elements of $\Theta\setminus\{\theta_0\}$.
	Let $c_\theta$ ($\theta_c$) be the action (state) associated with $\theta$ ($c$).
	Receiver $r$'s utility function reads: 
	\begin{equation*}
	u_r(\theta,c) = \begin{cases}
	\begin{array}{ll}
	1 &\textnormal{if } \theta=\theta_0 \textnormal{ and } c=\hat c^r\\
	0 &\textnormal{if } \theta = \theta_0 \textnormal{ and } c\neq \hat c^r\\
	\beta_r(c) &\textnormal{if } \theta \neq \theta_0 \textnormal{ and } c =c_\theta\\
	0 &\textnormal{if } \theta \neq \theta_0 \textnormal{ and } c\neq c_\theta
	\end{array}
	\end{cases}.
	\end{equation*}
	Let sender's utility be such that $f_\theta=f_0$, for each $\theta\neq\theta_0$, and $f_{\theta_0}=f$.
	We have that $f_\theta(\s)=0$ for each $\theta\in\Theta\setminus\{\theta_0\}$ and $\s\in S$.
	Then, there exists an optimal signaling scheme such that, in each state $\theta\neq\theta_0$, $\phi(\theta,\s_\theta)=1$, where $\s_\theta$ is a signal recommending $c_\theta$ to each receiver (from an argument analogous to Lemma~\ref{lemma:positive_prob}). 
	Now, an optimal signaling scheme can be computed by focusing on $\theta_0$ (\emph{i.e.}, we employ the aforementioned signaling scheme for any $\theta\neq\theta_0$)  via the following LP:
	\begin{subequations}\label{eq:equivalent}
		\begin{align}
		\max_{\substack{\phi(\theta_0,\cdot)\in \\ [0,1]^{|C\times R|}}} & \sum_{\s\in S}\phi(\theta_0,\s)f_{\theta_0}(\s)\\
		\nonumber \textnormal{s.t.  } & \sum_{\theta\in\Theta}\sum_{\s:s_r=c}\mu(\theta)\phi(\theta,\s)(u_r(\theta,c)-u_r(\theta,c'))\geq 0\\
		&\hspace{2.5cm}\forall r\in R,\forall c,c'\in C\label{eq:equivalent_ic}\\
		& \sum_{\s\in S}\phi(\theta_0,\s)=1.
		\end{align}
	\end{subequations}
	The incentive constraints~\ref{eq:equivalent_ic} are trivially satisfied when $c=\hat c^r$.
	Moreover, for each $c \neq \hat c^r$, the incentive constraints~\ref{eq:equivalent_ic} can be rewritten as follows: first, notice that it is enough to consider $c'=\hat c^r$.
	Then, for each $r\in R$ and $c\in C_r$, we obtain:
	\begin{multline*}
	\sum_{\s:s_r=c}\phi(\theta_0,\s)(u_r(\theta_0,c)-u_r(\theta_0,\hat c^r))\geq \\ u_r(\theta_c,\hat c^r)-u_r(\theta_c,c),
	\end{multline*}
	which can be rewritten as $\sum_{\s:s_r=c}\phi(\theta_0,s)\leq\beta_r(c)$. 
	The equivalence between LP~\ref{eq:subproblem_dual} and LP~\ref{eq:equivalent} easily follows.
\end{proof}

\anon*
\begin{proof}
	It is enough to provide an algorithm for the maximization problem in Theorem~\ref{th:general}.
	We need to solve $\max_{\s \in S}f(\s)+\sum_{r\in R}w_r(s_r)$.
	Since $f$ is anonymous, for any persuasive signal $\s$, $f$'s value is determined by the vector $\p=(\delta(\s,c_0),\ldots,\delta(\s,c_{\ell}))$.
	Let $P=\{\p=(k_0,\ldots,k_\ell)\in\mathbb{N}_0^{|C|}| \sum_{i=0}^{\ell} k_i=|R|\}$, and notice that $|P|=\binom{|R|+|C|-1}{|R|}$, which is polynomial in the input size once the $|C|$ has been fixed (see~\cite{Stanley2011}).
	In order to solve the maximization problem, we enumerate over all $\p \in P$. 
	Once $\p$ has been fixed, we are left with the following problem: $\max_{\s \in S}\sum_{r\in R} w_r(s_r)$, where $\s$ has to be such that $\delta(\s,c_i)=k_i$ for each $i\in\{0,\ldots,\ell\}$.
	Specifically, the optimal assignment of receivers to actions can be found with the following LP:
	\begin{subequations}
		\begin{align*}
		\max_{\chi\in\mathbb{R}_{+}^{|R\times C|}} & \sum_{(r,c)\in R\times C}\chi_r(c) w_r(c)\\
		\textnormal{ s.t. } & \sum_{r\in R}\chi_r(c_i) = k_i \hspace{.5cm}\forall i\in\{0,\ldots,\ell\}\\
		& \sum_{c\in C} \chi_r(c)=1 \hspace{.75cm}\forall r\in R.
		\end{align*}
	\end{subequations}
	We look for an integer solution of the problem, which always exists and can be found in polynomial time (see, e.g.,~\citep{orlin1997polynomial}).
	This is because the formulation is an instance of the \emph{maximum cost flow problem}, which is, in its turn, a variation of the \emph{minimum cost flow problem}.
	%
	%
	Once an integer solution has been found, an optimal solution of the original maximization problem is the signal obtained by assigning to each $r$ the action $c$ s.t. $\chi_r(c)=1$.
\end{proof}

\corollaryNeg*

\begin{proof}
	\textsf{PL-V} with two candidates is equivalent to \textsf{K-V} with $k^*=\lfloor\frac{|R|}{2}\rfloor+1$. We show that  \textsf{K-V} with arbitrary $k$ reduces to \textsf{K-V} with $k=k^*$. Theorem~\ref{th:negative} concludes the proof.
	
	We distinguish two cases: 
	i) Suppose $k>k^*$. We add $2k-|R|-1$ voters that prefer $c_1$ in any state. There are $|R^*|=2k-1$ voters and candidate $c_0$ has $k=\lfloor\frac{|R^*|}{2}\rfloor+1$ votes only if $k$ of the initial receivers vote for $c_0$.
	ii) Suppose $k< k^*$. We add $|R|+1-2k$ voters that prefer $c_0$ in any state. There are $|R^*|=2|R|+1-2k$ voters and candidate $c_0$ has $\lfloor\frac{|R^*|}{2}\rfloor+1=|R|-k+1$ votes only if $k$ of the initial receivers vote for $c_0$.
\end{proof}